\documentclass{article}
\author{David Bremner\footnote{Faculty of Computer Science, University of New Brunswick,\newline Fredericton, NB, Canada \newline Email: \{bremner, ra.shahsavari\} @unb.ca} \:\:\:\:\:\ Rasoul Shahsavarifar$^*$ \vspace{.5cm}}
\title{\textbf{Approximate Data Depth Revisited}}
\date{\today}

\usepackage{hyperref}
\usepackage{fullpage}
\usepackage{amsmath}
\usepackage{enumerate}
\usepackage{pifont}
\usepackage{polynom}
\usepackage{subfig}
\usepackage{amsfonts}
\usepackage[all]{xy}
\usepackage{graphicx}
\usepackage{titlesec}
\usepackage{color}
\usepackage{tabularx}
\usepackage{caption}
\usepackage{array}
\usepackage{pifont}
\usepackage{amssymb}
\usepackage[english]{babel}
\usepackage{amsthm}
\theoremstyle{definition}

\newtheorem{theorem}{Theorem}[section]
\newtheorem{conj}[theorem]{Conjecture}
\newtheorem{lemma}[theorem]{Lemma}

\DeclareMathOperator{\SkD}{\textit{SkD}}
\DeclareMathOperator{\HD}{\textit{HD}}

\begin{document}
\thispagestyle{empty}
\maketitle

\begin{abstract}
Halfspace depth and $\beta$-skeleton depth are two types of depth functions in nonparametric data analysis. The halfspace depth of a query point $q\in \mathbb{R}^d$ with respect to $S\subset\mathbb{R}^d$ is the minimum portion of the elements of $S$ which are contained in a halfspace which passes through $q$. For $\beta \geq 1$, the $\beta$-skeleton depth of $q$ with respect to $S$ is defined to be the total number of \emph{$\beta$-skeleton influence regions} that contain $q$, where each of these influence regions is the intersection of two hyperballs obtained from a pair of points in $S$. The $\beta$-skeleton depth introduces a family of depth functions that contain \emph{spherical depth} and \emph{lens depth} if $\beta=1$ and $\beta=2$, respectively. The main results of this paper include approximating the planar halfspace depth and $\beta$-skeleton depth using two different approximation methods. First, the halfspace depth is approximated by the $\beta$-skeleton depth values. For this method, two dissimilarity measures based on the concepts of \emph{fitting function} and \emph{Hamming distance} are defined to train the halfspace depth function by the $\beta$-skeleton depth values obtaining from a given data set. The goodness of this approximation is measured by a function of error values. Secondly, computing the planar $\beta$-skeleton depth is reduced to a combination of some range counting problems. Using existing results on range counting approximations, the planar $\beta$-skeleton depth of a query point is approximated in $O(n\;poly(1/\varepsilon,\log n))$, $\beta\geq 1$. Regarding the $\beta$-skeleton depth functions, it is also proved that this family of depth functions converge when $\beta \to \infty$. Finally, some experimental results are provided to support the proposed method of approximation and convergence of $\beta$-skeleton depth functions. 
\end{abstract}
\section{Introduction}
\label{sec:intro}
Data depth is a method to generalize the concept of rank in the univariate data analysis to the multivariate case. Data depth measures the centrality of a data point with respect to a dataset, and it gives a center-outward ordering of data points. In other words, applying a data depth on a dataset generates a partial ordered set (poset) of the data points. A poset is a set together with a partial ordering relation which is reflexive, antisymmetric and transitive. Over the last decades various notions of data depth such as \emph{halfspace depth} (Hotelling \cite{hotelling1990stability,small1990survey}; Tukey \cite{tukey1975mathematics}), \emph{simplicial depth} (Liu \cite{liu1990notion}) \emph{Oja depth} (Oja \cite{oja1983descriptive}), \emph{regression depth} (Rousseeuw and Hubert \cite{rousseeuw1999regression}), and others have been introduced in the area of non-parametric multivariate data analysis. These depth functions are different in application, definition, complexity of computations.  Among different notions of data depth, we focus on halfspace depth and a recently defined data depth named $\beta$-skeleton depth (Yang and Modarres \cite{yang2017beta}).
\\\\In 1975, Tukey generalized the definition of univariate median and defined the halfspace median as a point in which the halfspace depth is maximized, where the halfspace depth is a multivariate measure of centrality of data points. Halfspace depth is also known as Tukey depth or location depth. In general, the halfspace depth of a query point $q$ with respect to a given data set $S$ is the smallest portion of data points that are contained in a closed halfspace through $q$ \cite{bremner2008output,tukey1975mathematics}. The halfspace depth function has various properties such as vanishing at infinity, affine invariance, and decreasing along rays. These properties are proved in \cite{donoho1992breakdown}. Many different algorithms for the computation of halfspace depth in lower dimensions have been developed elsewhere \cite{bremner2008output,bremner2006primal,chan2004optimal,rousseeuw1998computing}. The bivariate and trivariate case of halfspace depth can be computed exactly in $O(n\log n)$ and $O(n^2 \log n)$ time \cite{rousseeuw1996algorithm,struyf1999halfspace}, respectively. However, computing the halfspace depth of a query point with respect to a data set of size $n$ in dimension $d$ is an NP-hard problem if both $n$ and $d$ are part of the input \cite {johnson1978densest}. Due to the hardness of the problem, designing efficient algorithms to compute (or approximate) the halfspace
depth of a point remains an interesting task in the research area of data depth \cite{afshani2009approximate,aronov2010approximate,chen2013absolute,har2011relative}.
\\\\In 2017, Yang and Modarres introduced a familly of depth functions called $\beta$-skeleton depth, indexed by a single parameter $\beta\geq 1$. The $\beta$-skeleton depth of a query point $q\in \mathbb{R}^d$ with respect to a given data set $S$ is defined as the portion of $\beta$-skeleton influence regions that contain $q$. The influence regions of $\beta$-skeleton depth are the multidimensional generalization of lunes in the definition of the \emph{$\beta$-skeleton graph} \cite{kirkpatrick1985framework}. A notable characteristic of the $\beta$-skeleton depth is related to its time complexity that grows linearly in the dimension $d$ whereas no polynomial algorithms (in the dimension) in higher dimensions are known for most other data depths. To the best of our knowledge, the current best algorithm for computing the $\beta$-skeleton depth in higher dimension $d$ is the straightforward algorithm which takes $\Theta(dn^2)$. The authors, in their previous work \cite{shahsavarifar2018computing}, improved this bound for the planar $\beta$-skeleton depth. They developed an $O(n^{3/2+\epsilon})$~algorithm for all values of $\beta\ge 1$, and a $\theta(n\log n)$ algorithm for the special case of $\beta=1$. Spherical depth (Elmore, Hettmansperger, and Xuan \cite{elmore2006spherical}) and lens depth (Liu and Modarres \cite{liu2011lens} can be obtained from $\beta$-skeleton depth by considering $\beta=1$ and $\beta=2$, respectively. It is proved that the $\beta$-skeleton depth function is monotonic, maximized at the center, and vanishing at infinity. The $\beta$-skeleton depth function is also orthogonally (affinely) invariant if the Euclidean (Mahalanobis) distance is used to construct the $\beta$-skeleton influence regions \cite{elmore2006spherical, liu2011lens,yang2014depth,yang2017beta}.
\\\\The concept of data depth is widely studied by statisticians and computational geometers. Some~directions that have been considered by researchers include defining new depth functions, improving the complexity of computations, computing both exact and approximate depth values, and computing depth functions in higher dimensions. Two surveys by Aloupis \cite{aloupis2006geometric} and Small \cite{small1990survey} can be referred as overviews of data depth from a computational geometer's and a statistician's point of view, respectively.
\\\\In this paper, different methods are presented to approximate the halfspace and $\beta$-skeleton depth functions. Computing the $\beta$-skeleton depth is reduced to a combination of range counting problems. Using different range counting~approximations in \cite{aronov2008approximating,har2011geometric,shaul2011range}, the planar $\beta$-skeleton depth ($\beta\geq 1$) of a given point is approximated  in $O(n\:poly(1/\varepsilon,\log n))$ query time. Furthermore, we propose an approximation technique to approximate the halfspace depth using the $\beta$-skeleton depth. In this method, two dissimilarity measures based on the concepts of \emph{fitting function} and \emph{Hamming distance} are defined to train the halfspace depth function by the $\beta$-skeleton depth values obtaining from a given data set. The goodness of approximation can be measured by the sum of square of error values. We also show that $\beta$-skeleton depth functions converge when $\beta\rightarrow \infty$. Finally, some experimental results are provided regarding our proposed method of approximation.

\section{Halfspace Depth}
\paragraph{Definition:} The halfspace depth of a query point $q\in \mathbb{R}^d$ with respect to a given data set $S=\{x_1,...,x_n\} \subseteq \mathbb{R}^d$ is defined as the minimum portion of points of $S$ contained in any closed halfspace that has $q$ on its boundary. Using the notation of $\HD (q;S)$, the above definition can be presented by~\eqref{eq:halfspacedef}.
\begin{equation}
\label{eq:halfspacedef}
\HD (q;S)= \frac{2}{n}\min \{|S \cap H|: H\in \mathbb{H} , q\in H \},
\end{equation}
where $2/n$ is the normalization factor\footnote{Instead of the normalization factor $1/n$ which is common in literature, we use the normalization factor $2/n$ in order to let the depth of $1$ to be achievable.},  $\mathbb{H}$ is the class of all closed halfspaces in $\mathbb{R}^d$ that pass through $q$, and $|S \cap H|$ denotes the number of points within the intersection of $S$ and $H$. As illustrated in Figure~\ref{fig:halfspace-example}, $\HD(q_1;S)=6/13$ and $\HD(q_1;S)=0$, where $S$ a given set of points in the plane and $q_1,q_2$ are two query points not in $S$.

\begin{figure}[!ht]
  \centering
    \includegraphics[width=0.5\textwidth]{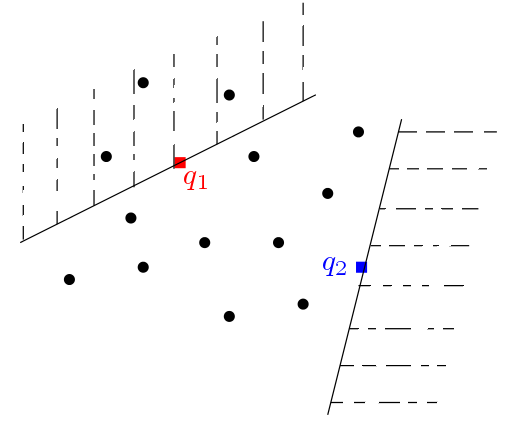}
  \caption{Two examples of halfspace depth in the plane}
  \label{fig:halfspace-example}
\end{figure}
 
\section{$\beta$-skeleton Depth}
\paragraph{Definition:} For $1 \leq \beta \leq \infty$, the $\beta$-skeleton influence region of $x_i$ and $x_j$ ($S_{\beta}(x_i, x_j)$) is defined as follows:
\begin{equation}
\label{eq:beta-influence}
S_{\beta}(x_i, x_j)= B(c_i,r) \cap B(c_j,r),
\end{equation}
where $r=\frac{\beta}{2}\Vert x_i - x_j\Vert$, $c_i=\frac{\beta}{2}x_i + (1-\frac{\beta}{2})x_j$,  and $c_j=(1-\frac{\beta}{2})x_i + \frac{\beta}{2}x_j$.
\\\\In the case of $\beta=\infty$, the $\beta$-skeleton influence region is well defined, and it is a slab defined by two halfspaces. Figure~\ref{fig:betasphericallens} shows the $\beta$-skeleton influence regions for different values of $\beta$.

\begin{figure}[!ht]
  \centering
    \includegraphics[width=0.5\textwidth]{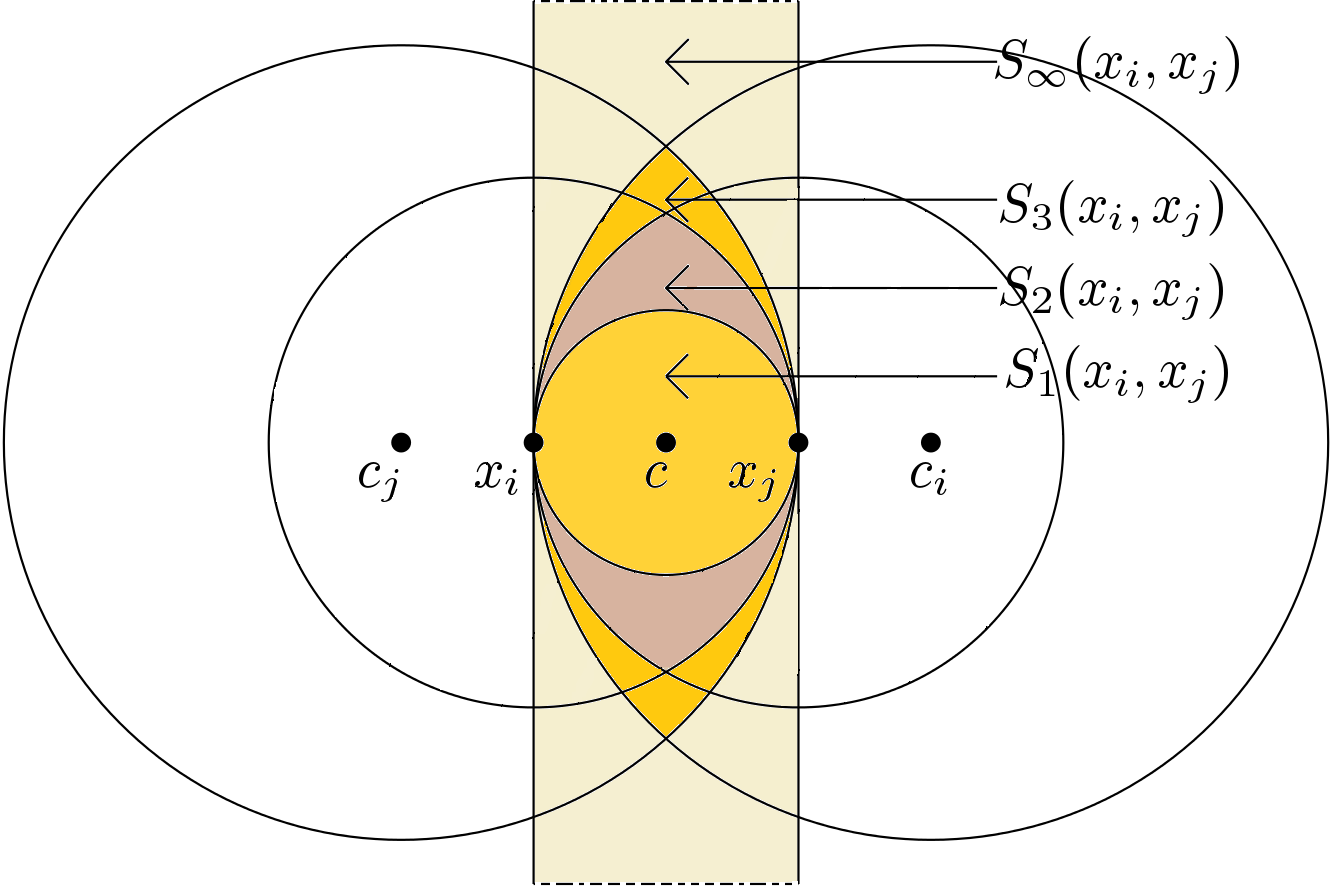}
  \caption{The $\beta$-skeleton influence regions defined by $x_i$ and $x_j$ for $\beta$=1, 2, 3, and $\beta=\infty$, where $c=\frac{x_i+x_j}{2}$, $c_i=\frac{3}{2}x_i+(1-\frac{3}{2})x_j$, and $c_j=(1-\frac{3}{2})x_i+\frac{3}{2}x_j$}
  \label{fig:betasphericallens}
\end{figure}

\paragraph{Definition:}Let $S=\{x_1, ..., x_n\}$ be a set of points in $\mathbb{R}^d$. For the parameter $1\leq \beta\leq \infty$, the $\beta$-skeleton depth of a query point $q \in \mathbb{R}^d$ with respect to $S$, is defined as a proportion of the $\beta$-skeleton influence regions of $S_{\beta}(x_i, x_j), 1\leq i < j \leq n$ that contain $q$. Using the indicator function $I$, this definition can be represented by Equation \eqref{eq:beta}.
\begin{equation}
\label{eq:beta}
\SkD_{\beta}(q;S)= \frac{1}{{n \choose 2}}\sum_{1\leq i<j\leq n} {I(q \in S_{\beta}(x_i, x_j)})
\end{equation}
\\It can be verified that $q \in S_{\beta}(x_i, x_j)$ is equivalent to the inequality of $\frac{\beta}{2}\Vert x_i - x_j \Vert \geq \max \{ \Vert q - c_i \Vert , \Vert q - c_j \Vert\}$, where $\Vert c_i - c_j \Vert=\Vert (1-\beta)(x_i-x_j) \Vert = (\beta - 1)\Vert x_i - x_j \Vert$ for $\beta \geq 1$. The straightforward algorithm for computing the $\beta$-skeleton depth of $q \in \mathbb{R}^d$ takes $\Theta(d n^2)$ time because the above inequality should be checked for all $1\leq i,j\leq n$.
\section{Dissimilarity Measures}
In this section, two different types of dissimilarity measures for depth functions are introduced.
\subsection{Fitting Functions and Dissimilarity Measures}
To determine the dissimilarity between two vectors $U=(u_1 ,..., u_n)$ and $V=(v_1, ..., v_n)$, the idea of fitting functions can be applied. Considering the goodness measures of fitting functions, assume that $f$ is the best function fitted to $U$ and $V$ which means that $u_i=f(v_i)\pm \delta_i$. Let $\xi_i=u_i-\overline{U}$, where $\overline{U}$ is the average of $u_i$ ($1\leq i\leq n$). We define the dissimilarity measure between $U$ and $V$ ($d_E(U,V)$) to be a function of $\delta_i$ and $\xi_i$ as follows:

\begin{equation}
\label{eq:e-distance}
d_E(U,V)=1-r^2,
\end{equation}
\\where $r^2$ is the coefficient of determination \cite{shafer2012beginning} which is defined by:
\begin{equation}
\label{eq:r-squared}
r^2=\dfrac{\sum\limits_{i=1}^{n}(\xi_i^2-\delta_i^2)}{\sum\limits_{i=1}^{n}\xi_i^2}.
\end{equation}  
\\Since $r^2\in[0,1]$, $d_E(U,V)\in [0,1]$. A smaller value of $d_E(U,V)$ represents more similarity between $U$ and~$V$. 
        
\paragraph{Definition:} Let $S=\{x_1,...,x_n\}$ be a finite set. It is said that $P=(S,\preceq)$ is a partially ordered set (poset) if $\preceq$ is a partial order relation on $S$, that is, for all $x_i,x_j,x_t \in S$: $(a)\; x_i \preceq x_i $; $(b)\; x_i \preceq x_j$ and $x_j \preceq x_t$ implies that $x_i \preceq x_t$; $(c)\; x_i \preceq x_j$ and $x_j \preceq x_i$ implies that $x_i \equiv_p x_j$, where $\equiv_p$ is the corresponding equivalency relation.
\\Poset $P=(S,\preceq)$ is called a chain if any two elements of $S$ are comparable, i.e., given $x_i ,x_j\in S$, either $x_i \preceq x_j $ or $x_j \preceq x_i $. If there is no comparable pair among the elements of $S$, the corresponding post is an anti chain. Figure~\ref{fig:poset} illustrates different posets with the same elements.

\begin{figure}[!ht]
  \centering
    \includegraphics[width=0.5\textwidth]{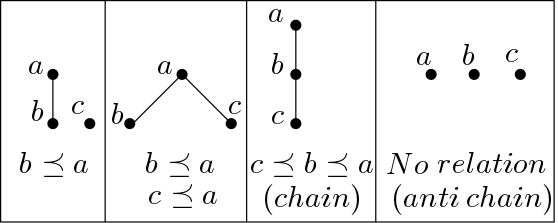}
  \caption{Different posets and relations among their elements}
  \label{fig:poset}
\end{figure}

\subsection{Dissimilarity Measures Between two Posets}
The idea of defining the following distance comes from the proposed structural dissimilarity measure between posets in \cite{fattore2014measuring}. Let $\mathbb{P}=\{P_t=(S,\preceq_t)\vert t\in \mathbb{N}\}$ be a finite set of posets, where $S=\{x_1, ..., x_n\}$. For $P_k\in\mathbb{P}$ we define a matrix $M^k_{n\times n}$ by:
\[
M^k_{ij}= \begin{cases}
1\;\;\;\;\;x_i \preceq_k x_j \\
0\;\;\;\;\; \text{otherwise.}
\end{cases}
\]
We use the notation of $d_c(P_f,P_g)$ to define a dissimilarity between two posets $P_f,P_g\in \mathbb{P}$ as follows:

\begin{equation}
\label{eq:h-distance}
d_c(P_f,P_g)=\dfrac{\sum\limits_{i=1}^{n}\sum\limits_{j=1}^{n} \vert M^f_{ij}-M^g_{ij} \vert}{n^2-n}
\end{equation}
It can be verified that $d_c(P_f,P_g)\in [0,1]$, where the closer value to $1$ means the less similarity between $P_f$ and $P_g$. This measure of similarity is a metric on $\mathbb{P}$ because for all $P_f,P_g,P_h \in \mathbb{P}$,
\begin{itemize}
\item $d_c(P_f,P_g)\geq 0$
\item $d_c(P_f,P_g)=0 \Leftrightarrow P_f=P_g$
\item $d_c(P_f,P_g)=d_c(P_g,P_f)$
\item $d_c(P_f,P_h)\leq d_c(P_f,P_g)+ d_c(P_g,P_h).$
\end{itemize}
Proving these properties is straightforward. The proof of last property which is less trivial can be found in Appendix (see Lemma \ref{lm:metric-property}).

\section{Approximation of Halfspace Depth}
\label{sec:approx-half-space}
Due to the difficulty of computing the halfspace depth in dimension $d>3$, many authors have tried to approximate the halfspace depth. We propose a method to approximate the halfspace depth using another depth function. Among all depth functions, the $\beta-$skeleton depth is chosen to approximate the halfspace depth because it is easy to compute and its time complexity, i.e. $\Theta(dn^2)$, grows linearly in higher dimension $d$. 

\subsection{Approximation of Halfspace Depth and Fitting Function}
Suppose that $S=\{x_1, ...,x_n\}$ is a set of data points. By choosing some subsets of $S$ as training samples, we consider the problem of learning the halfspace depth function using the $\beta-$skeleton depth values. Finally, by applying the cross validation techniques in machine learning, the best function $f$ can be obtained such that $\HD(x_i;S)=f(\SkD_{\beta}(x_i;S))\pm \delta_i$. The function $f$ can be considered as an approximation function for halfspace depth, where the value of $d_E(\HD,\SkD_{\beta})$ that can be computed using Equation \eqref{eq:e-distance} is the error of approximation.

\subsection{Approximation of Halfspace Depth and Poset Dissimilarity}
In some applications, the structural ranking among the elements of $S$ is more important than the depth value of single points. Let $S=\{x_1, ..., x_n\}$ be a set of points and $D$ be a depth function. Applying $D$ on $x_i$ with respect to $S$ generates a poset (in particular, a chain). In fact, $P_D=(D(x_i;S),\leq)$ is a chain because for every $x_i, x_j \in S$, the values of $D(x_i;S)$ and $D(x_j;S)$ are comparable. For halfspace depth and $\beta$-skeleton depth, their dissimilarity measure of rankings can be obtained by Equation \eqref{eq:h-distance} as follows:
\[
d_c(\HD,\SkD_{\beta})=\dfrac{\sum\limits_{i=1}^{n}\sum\limits_{j=1}^{n} \vert M^{\HD}_{ij}-M^{\SkD_{\beta}}_{ij} \vert}{n^2-n}.
\]
The smaller value of $d_c(\HD,\SkD_{\beta})$, the more similarity between $\HD$ and $\SkD_{\beta}$ in ordering the elements of $S$.
\\\\ In both of the above approximation methods, any other depth function can be considered instead of $\beta$-skeleton depth to approximate the halfspace depth.    

\begin{conj}
\label{conj:1}
For two depth functions $D_1$ and $D_2$, the small value of $d_c(D_1,D_2)$ implies the small value of $d_E(D_1,D_2)$ and vice versa. 
\end{conj}

\section{Approximation of $\beta$-skeleton Depth}
The convergence of $\beta$-skeleton depth functions when $\beta\to\infty$ is investigated in this section. Furthermore, the problem of planar $\beta$-skeleton depth is reduced to a combination of disk and halfspace range counting problems. This reduction is applied to approximate the planar $\beta$-skeleton depth using the range counting approximation.
\subsection{Convergence of $\beta$-skeleton Depth Functions}
The following theorem helps understand the definition of $\SkD_{\infty}$ in Equation \eqref{eq:beta}.
\begin{theorem}
\label{thrm:beta-converge}
If $\beta\to \infty$, all $\beta$-skeleton depth functions converge to $\SkD_{\infty}$. In other words, for data set $S=~\{x_1, ..., x_n\}$ and query point $q$,
\begin{equation}
\label{eq:beta-converge}
\lim_{\beta \to \infty}\frac{\SkD_{\beta}(q;S)}{\SkD_{(\beta+1)}(q;S)}=1.
\end{equation}
\end{theorem}
\begin{proof}
Referring to \eqref{eq:beta}, the definition of $\beta$-skeleton depth, it is enough to prove that if $\beta \to \infty$,

\begin{equation}
\label{eq:influence-converge}
\forall x_i,x_j\in S;\; S_{\beta}(x_i,x_j)=S_{(\beta+1)}(x_i,x_j).
\end{equation}
It is proved that $S_\beta(x_i,x_j)\subseteq S_{(\beta+1)}(x_i,x_j)$, and $u_{\beta}=u_{(\beta+1)}$ if $\beta\to \infty$ (see Lemma \ref{lm:beta-influence-containment} and Lemma \ref{lm:lim-delta} in Appendix). Since the influence regions of $\beta$-skeleton depth functions are closed and convex, the proof is complete if
\begin{equation}
\label{eq:influence-area-limit}
\lim_{\beta \to \infty}\frac{A(S_{\beta i j})}{A(S_{(\beta+1) i j})}=1,
\end{equation}
where $A(S_{\beta i j})$ is the area of $S_{\beta}(x_i,x_j)$. It can be verified that $A(S_{\beta i j})$ is equal to $\theta r^2-da$, where $d=(\beta-1)l$, $a=(\l/2)\sqrt{(2\beta-1)}$, $r=\beta l/2$, $\theta=\cos^{-1}((\beta-1)/\beta)$, and $l=d(x_i,x_j)$. See Figure \ref{fig:beta-converge}.

\begin{figure}[!ht]
  \centering
    \includegraphics[width=0.5\textwidth]{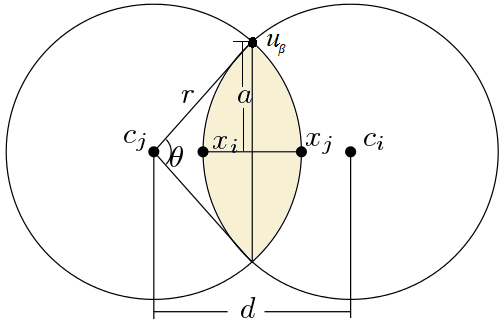}
  \caption{An illustration of $d$, $a$, $r$, $l$, and $\theta$.}
  \label{fig:beta-converge}
\end{figure}
\[
\frac{A(S_{\beta i j})}{A(S_{(\beta+1) i j})}=\frac{\beta^2\cos^{-1}[(\beta-1)/\beta]-(\beta-1)\sqrt{2\beta-1}}{(\beta+1)^2\cos^{-1}[\beta/(\beta+1)]-\beta\sqrt{2\beta+1}}
\]
Equation \eqref{eq:influence-area-limit} is proved because 
\[
\lim_{\beta \to \infty}\cos^{-1}(\frac{\beta-1}{\beta})=\lim_{\beta \to \infty}\cos^{-1}(\frac{\beta}{\beta+1})=0.
\]
\end{proof}
\subsection{$\varepsilon-$approximation of Planar $\beta$-skeleton Depth}
In this section, we prove that for $S=\{x_1,...,x_n\}\subset\mathbb{R}^2$ and $q\in \mathbb{R}^2$, computing $\SkD_\beta(q;S)$ is equivalent with at most $3n$ semialgebraic range counting problems. This result can be applied to approximate the planar $\beta$-skeleton depth in $O(n\:poly(1/\varepsilon,\log n))$ query time with $O(n\:poly(1/\varepsilon,\log n))$ storage and the same bound for preprocessing time.  
\\\\Given a set $S=\{x_1, ..., x_n\}\subset \mathbb{R}^2$, a semialgebraic range $\tau$ with constant description complexity, and a parameter $\varepsilon >0$, $S$ can be preprocessed into a data structure that helps to efficiently compute an approximated count $\eta_\tau$ which satisfies the \textit{$(\varepsilon,\tau)$-approximation} given by:
\begin{equation*}
(1-\varepsilon) \vert S \cap \tau \vert\leq \eta_\tau \leq (1+\varepsilon) \vert S \cap \tau \vert.
\end{equation*}
Considering $S$ and $\varepsilon$ as introduced above, a short list of latest results for approximation of some semialgebraic range counting problems in $\mathbb{R}^2$ is as follows.  
\begin{itemize}
\item \textbf{Halfspace:}
With $O(n\:poly(1/\varepsilon,\log n))$ preprocessing time, one can construct a data structure of size $O(n\:poly(1/\varepsilon,\log n))$ such that for a given halfspace $\hbar$, it outputs a number $\eta_\hbar$ that satisfies the $(\varepsilon,\hbar)$-approximation with the query time $O(poly(1/\varepsilon,\log n))$. For points in $\mathbb{R}^3$, the same results can be obtained \cite{har2011geometric}.
\item \textbf{Disk:} By standard lifting of planar  points to the paraboloid in $\mathbb{R}^3$, one can reduce a disk range query in the plane to a halfspace range query in three dimensions \cite{har2011geometric}. 
\item \textbf{Circular cap:} A circular cap is the larger part of a circle cut by a line. In near linear time, $S$ can be preprocessed into a linear size data structure that returns the emptiness result of any circular cap query in $O(poly \log n)$. With $O(\varepsilon^{-2} T(n)\log n)$ preprocessing time and $O(\varepsilon^{-2} S(n)\log n)$ storage, one can construct a data structure such that for a given circular cap $\rho$, it returns in $O(\varepsilon^{-2}Q(n)\log n)$ time, a number $\eta_\rho$ that satisfies the $(\varepsilon,\rho)$-approximation. Both of $T(n)$ and $S(n)$ are near linear and $Q(n)$ is the emptiness query of $\rho$ \cite{shaul2011range}.
\end{itemize}

\paragraph{Definition:}
For an arbitrary non-zero point $a \in \mathbb{R}^2$  and parameter $\beta \geq 1$, $\ell(p)$ is a line that is perpendicular to $\overrightarrow{a}$ at the point $p=p(a,\beta)={(\beta -1)a}/{\beta}$. This line forms two halfspaces $H_o(p)$ and $H_a(p)$. The one that includes the origin is $H_o(p)$ and the other one that includes $a$ is $H_a(p)$.

\paragraph{Definition:}For a disk $B(c,r)$ with the center $c=c(a,\beta)={\beta a}/{(2(\beta -1))}$ and radius $r=\Vert c \Vert$, $B_o(c,r)$ is the intersection of $H_o(p)$ and $B(c,r)$,  and $B_a(c,r)$ is the intersection of $H_a(p)$ and $B(c,r)$, where $\beta>1$ and $a$ is an arbitrary non-zero point in $\mathbb{R}^2$.
\\\\Figure \ref{fig:halfspace-balls} is an illustration of these definitions for different values of parameter $\beta$.

\begin{figure*}[!ht]
  \centering
    \includegraphics[width=0.78\textwidth]{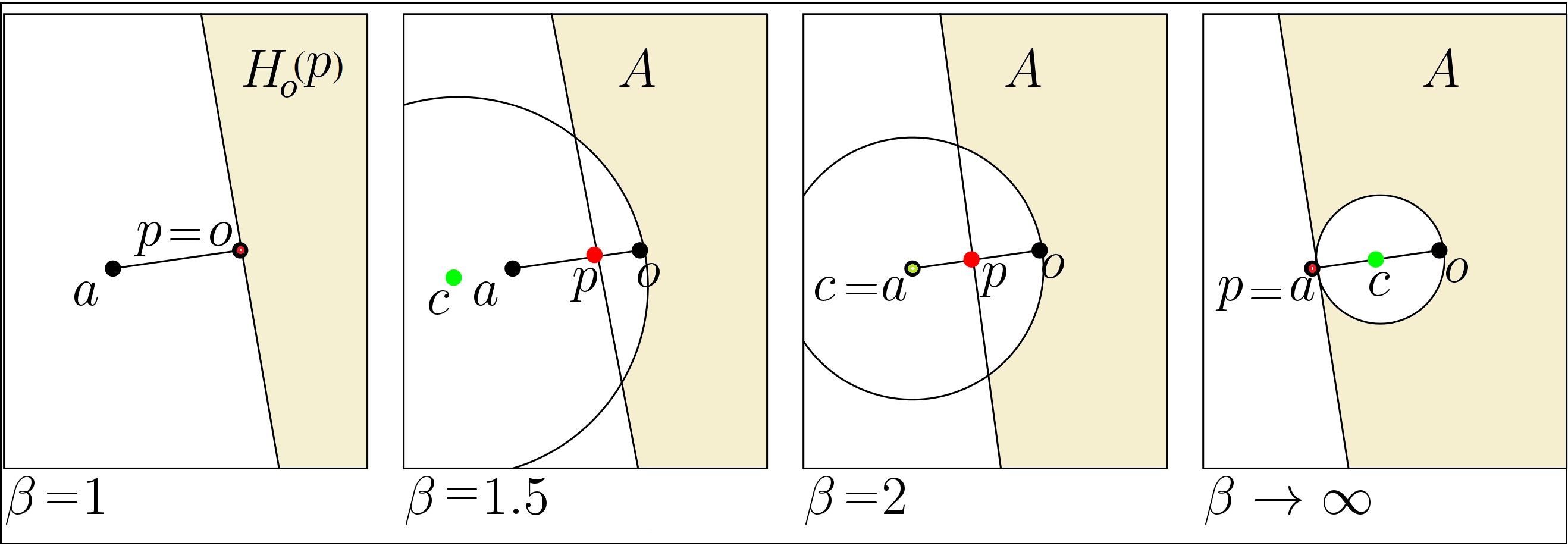}
  \caption{The $H_o(p)$ and $B(c,r)$ defined by $a\in \mathbb{R}^2$ for $\beta=1,\;1.5,\;2,\;\text{and}\;\beta\rightarrow\infty$, where $A=H_o(p)\setminus \{intB_o(c,r)\}$}
\label{fig:halfspace-balls}
\end{figure*}

\begin{theorem}
For arbitrary non-zero points $a$, $b$ in $\mathbb{R}^2$ and parameter $\beta>1$, $b \in H_o(p)\setminus \{int B_o(c,r)\}$ if and only if the origin $O=(0,0)$ is contained in $S_{\beta}(a,b)$, where $c={\beta a}/{(2(\beta -1))}$, $r=\Vert c \Vert$, and $p={(\beta -1)a}/{\beta}$.
\label{thrm:q-skeleton}
\end{theorem}

\begin{proof} 
First, we show that $B_o(c,r)$ is a well-defined set meaning that $\ell(p)$ intersects $B(c,r)$. We compute $d(c,\ell(p))$, the distance of $c$ from $\ell(p)$, and prove that this value is not greater than $r$. It can be verified that $d(c,\ell(p))= d(c,p)$. Let $k={\beta}/{(2(\beta -1))}$; the following calculations complete this part of the proof.
\begin{align*}
d(c,p)&=d(\frac{\beta a}{2(\beta -1)},\frac{(\beta -1)a}{\beta})= d(ka,\frac{1}{2k}a)\\&=(k-\frac{1}{2k})\sqrt{({a_x}^2+{a_y}^2)}=(\frac{2k^2-1}{2k})\Vert a \Vert\\& \leq \frac{2k^2}{2k}\Vert a \Vert= k\Vert a \Vert = r
\end{align*}
\\We recall the definition of $\beta$-influence region given by $S_{\beta}(a,b)=B(c_a,\frac{\beta}{2}\Vert a-b \Vert) \cap B(c_b,\frac{\beta}{2}\Vert a-b \Vert)$, where $c_a=\frac{\beta}{2}a +(1-\frac{\beta}{2})b$ and $c_b=\frac{\beta}{2}b +(1-\frac{\beta}{2})a$. Using this definition, following equivalencies can be derived from $O \in S_{\beta}(a,b)$.
\[O \in S_{\beta}(a,b) \Leftrightarrow \frac{\beta \Vert a-b \Vert}{2} \geq max\{\Vert c_a \Vert , \Vert c_b \Vert\}\Leftrightarrow\]
\[\beta \Vert a-b \Vert \geq max\{\Vert \beta(a-b)+2b\Vert , \Vert \beta(b-a)+2a \Vert\} \Leftrightarrow\]
\[\beta ^2 \Vert a-b \Vert ^2 \geq max\{\Vert \beta(a-b)+2b\Vert ^2 , \Vert \beta(b-a)+2a \Vert ^2\} \]
\[\Leftrightarrow 0 \geq max\{b^2(1-\beta)+\beta \overrightarrow{a}.\overrightarrow{b} , a^2(1-\beta)+\beta \overrightarrow{a}.\overrightarrow{b}\}.\]
\\By solving these inequalities for $(\beta -1)/\beta$ which is equal to $1/2k$, we have:  
\begin{equation}
\label{eq:equivalencies}
\frac{1}{2k} \geq max \left\{ \frac{\overrightarrow{a}.\overrightarrow{b}}{\Vert a \Vert ^2}, \frac{\overrightarrow{a}.\overrightarrow{b}}{\Vert b \Vert ^2} \right\}.
\end{equation}
For a fixed point $a$, the inequalities in Equation \eqref{eq:equivalencies} determine one halfspace and one disk given by \eqref{eq:halfspace} and \eqref{eq:disk}, respectively.
\begin{equation}
\label{eq:halfspace}
\frac{1}{2k} \geq \frac{\overrightarrow{a}.\overrightarrow{b}}{\Vert a \Vert ^2} \Leftrightarrow \overrightarrow{a}.\overrightarrow{b} \leq\frac{1}{2k} \Vert a \Vert ^2.
\end{equation}
\begin{equation}
\label{eq:disk}
\begin{split}
&\frac{1}{2k} \geq \frac{\overrightarrow{a}.\overrightarrow{b}}{\Vert b \Vert ^2} \Leftrightarrow  b^2-2k\overrightarrow{a}.\overrightarrow{b} \geq 0 \Leftrightarrow 
b^2-2k\overrightarrow{a}.\overrightarrow{b}+k^2a^2 \geq k^2a^2  \Leftrightarrow  \left( b- ka\right)^2 \geq \left( k \Vert a \Vert\right)^2.
\end{split}
\end{equation}
The proof is complete because for a point $a$, the set of all points $b$ containing in the feasible region defined by Equations \eqref{eq:halfspace} and \eqref{eq:disk} is equal to \\$H_o(p)\setminus \{intB_o(c,r)\}$.  
\end{proof}
\paragraph{Note~1:} It can be verified that for the value of $\beta=2+\sqrt{2}$, the given halfspace in \eqref{eq:halfspace} passes through the center of the given disk in \eqref{eq:disk}. See Lemma \ref{lm:halfspace-circle-center} in Appendix.
\\\\For $S=\{x_1, ..., x_n\}\subset \mathbb{R}^2$ and $x_i\in S$, let $\hbar(x_i)$ be the given halfspace by \eqref{eq:halfspace}, and $B(x_i)$ be the given disk by \eqref{eq:disk}. From Theorem \ref{thrm:q-skeleton}, it can be deduced that we need to do  
\begin{itemize}

\item $n$ halfspace range counting approximations to approximate $\SkD_1(q;S)$ because
\begin{equation*}
\SkD_1(q;S)=\frac{1}{2}\sum_{i=1}^n\vert \hbar(x_i) \vert.
\end{equation*}

\item $n$ halfspace and $n$ disk range counting approximations to approximate $\SkD_{\infty}(q;S)$ because
\begin{equation*}
\SkD_{\infty}(q;S)=\frac{1}{2}\sum_{i=1}^n(\vert \hbar(x_i) \vert-\vert B(x_i) \vert).
\end{equation*}

\item $n$ halfspace, $n$ disk, and $n$ circular cap range counting approximations to approximate $\SkD_{\beta}(q;S)$, for $1<\beta <2+\sqrt{2}$, because
\begin{equation*}
\SkD_{\beta}(q;S)=\frac{1}{2}\sum_{i=1}^n(\vert \hbar(x_i) \vert-\vert B(x_i) \vert + \vert \rho(x_i) \vert),
\end{equation*}
 where $\rho(x_i)$ is the circular cap obtained from $b(x_i)$ cut by $\hbar(x_i)$.
 
\item $n$ halfspace and $n$ circular cap range counting approximations to approximate $\SkD_{\beta}(q;S)$, for $2+~\sqrt{2}\leq\beta <\infty$, because
\begin{equation*}
\SkD_{\beta}(q;S)=\frac{1}{2}\sum_{i=1}^n(\vert \hbar(x_i) \vert- \vert \rho(x_i) \vert).
\end{equation*}

\end{itemize}

\section{Experimental Results}
In this section some experimental results are provided to support Section \ref{sec:approx-half-space} and Theorem \ref{thrm:beta-converge}. We compute the planar halfspace depth and planar $\beta$-skeleton depth of $q\in Q$ with respect to $S$ for different values of $\beta$, where $Q$ with the size of $1000$ and $S$ with the size of $2500$ are two sets of randomly generated points (double precision floating points) within the square $\{(x,y)| x,y \in [-10,10]\}$. The results of our experiments are summarized in Table~\ref{tbl:results-fitting-function} and its corresponding figures provided in Appendix. The columns $2-4$ of Table \ref{tbl:results-fitting-function} show that the halfspace depth can be approximated by a quadratic function of the $\beta$-skeleton depth with a relatively small value of $d_E(\SkD,\HD)\approxeq 0.003$. In particular, $\HD\approxeq4.3(\SkD_{\beta})^2-2.8(\SkD_{\beta})+0.47$ if $\beta\to \infty$. Columns $5-7$ include the experimental results to support Theorem \ref{thrm:beta-converge}.

\section{Acknowledgment}
The authors would like to thank Joseph O'Rourke and other attendees of CCCG2017 for inspiring discussions regarding the relationships between depth functions and posets. 

%---------------------------- Bibliography -------------------------------

% Please add the contents of the .bbl file that you generate,  or add bibitem entries manually if you like.
% The entries should be in alphabetical order

\newpage
\section{Appendix}
\begin{lemma}
\label{lm:metric-property}
For posets $A,B$, $C$ in $\mathbb{P}=\{P_t=(S,\preceq_t)\vert t\in \mathbb{N}\}$, $d_c(A,B)\leq d_c(A,C)+ d_c(C,B),$ where 
\[
d_c(A,B)=\dfrac{\sum\limits_{i=1}^{n}\sum\limits_{j=1}^{n} \vert M^A_{ij}-M^B_{ij} \vert}{n^2-n}
\]
and \[
M^k_{ij}= \begin{cases}
1\;\;\;\;\;x_i \preceq_k x_j \\
0\;\;\;\;\; \text{otherwise.}
\end{cases}
\]
\end{lemma}
\begin{proof}
\begin{align*}
d_c(P_f,P_h)&=\dfrac{\sum\limits_{i=1}^{n}\sum\limits_{j=1}^{n} \vert M^f_{ij}-M^h_{ij} \vert}{n^2-n} \\&= \dfrac{\sum\limits_{i=1}^{n}\sum\limits_{j=1}^{n} \vert (M^f_{ij}-M^g_{ij})+(M^g_{ij}-M^h_{ij})\vert}{n^2-n}
\\&\leq \dfrac{\sum\limits_{i=1}^{n}\sum\limits_{j=1}^{n} (\vert M^f_{ij}-M^g_{ij}\vert +\vert M^g_{ij}-M^h_{ij}\vert)}{n^2-n}
\\&= \dfrac{\sum\limits_{i=1}^{n}\sum\limits_{j=1}^{n} \vert M^f_{ij}-M^g_{ij} \vert}{n^2-n} + \dfrac{\sum\limits_{i=1}^{n}\sum\limits_{j=1}^{n} \vert M^g_{ij}-M^h_{ij} \vert}{n^2-n}\\&=d_c(P_f,P_g)+d_c(P_g,P_h).
\end{align*}
\end{proof}
\begin{lemma}
\label{lm:beta-influence-containment}
For $\beta'>\beta\geq 1$ and $a,b\in \mathbb{R}^2$, $S_{\beta}(a,b)\subseteq S_{\beta'}(a,b)$, where $S_{\beta}(a,b)$ is the intersection of two disks $B(C_{ab\beta},R_{ab\beta})$ and $B(C_{ba\beta},R_{ba\beta})$, $C_{ab\beta}=(\beta/2)(a-b)+b$, and  $R_{ab\beta}=(\beta/2)d(a,b)$.
\end{lemma}
\begin{proof}
To prove that $B(C_{ab\beta},R_{ab\beta}) \cap B(C_{ba\beta},R_{ba\beta})$ is a subset of $B(C_{ab\beta'},R_{ab\beta'})\cap B(C_{ba\beta'},R_{ba\beta'})$, it is enough to prove $B(C_{ab\beta},R_{ab\beta})\subseteq B(C_{ab\beta'},R_{ab\beta'})$ and $B(C_{ba\beta},R_{ba\beta})\subseteq B(C_{ba\beta'},R_{ba\beta'})$. We only prove the first one, and the second one can be proved similarly.
Suppose that $\beta<\beta'=\beta+\varepsilon; \varepsilon>0$. It is trivial to check that two disks $B(C_{ab\beta},R_{ab\beta})$ and $B(C_{ab\beta'},R_{ab\beta'})$ meet at $b$. See Figure~\ref{fig:two-beta-values}. Let $t\neq b$ be an extreme point of $B(C_{ab\beta},R_{ab\beta})$. This means that

\begin{align*}
d(t,C_{ab\beta})&=R_{ab\beta} \Leftrightarrow d(t,\frac{\beta(a-b)}{2}+b)=\frac{\beta(d(a,b))}{2}\\&\Leftrightarrow \left\vert \frac{\beta(a-b)}{2}+(b-t)\right\vert^2=\left(\frac{\beta(a-b)}{2}\right)^2\\&
\Leftrightarrow \vert b-t\vert^2-\beta(b-a)\cdot (b-t)=0
\end{align*}
This means $(b-a)\cdot (b-t)\geq 0$. Hence,
\begin{align*}
&\vert b-t\vert^2-\beta(b-a)\cdot (b-t)=0 \Leftrightarrow\\
&\vert b-t\vert^2-(\beta+\varepsilon)(b-a)\cdot (b-t)<0 \Leftrightarrow\\
&\vert b-t\vert^2-\beta'(b-a)\cdot (b-t)<0 \Leftrightarrow\\
&d(t,C_{ab\beta'})-R_{ab\beta'}<0 
\end{align*}
The last inequality means that $t$ is an interior point of $B(C_{ab\beta'},R_{ab\beta'}).$ 
\end{proof}

\begin{figure}[!ht]
  \centering
    \includegraphics[width=0.7\textwidth]{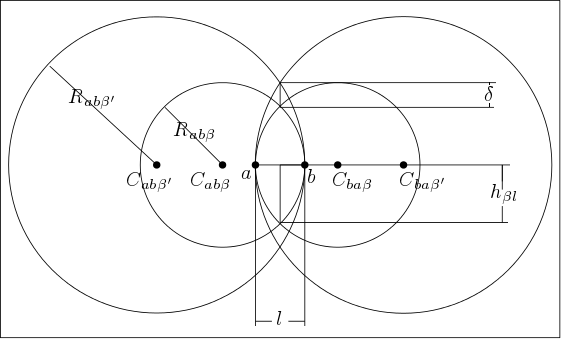}
  \caption{$S_{\beta}(a,b)$ and $S_{\beta'}(a,b)$ .}
  \label{fig:two-beta-values}
\end{figure}

\begin{lemma}
\label{lm:lim-delta}
For $a,b\in \mathbb{R}^2$, 
\begin{equation}
\label{eq:lim-delta}
\lim_{\beta \to \infty}\delta= \lim_{\beta \to \infty}(h_{(\beta+1)\l}-h_{\beta\l})=0,
\end{equation}
where $h_{\beta\l}=(\l/2)\sqrt{(2\beta-1)}$. See Figure \ref{fig:two-beta-values}.

\begin{proof}
Instead of proving \eqref{eq:lim-delta}, we prove its equivalent form as follows:
\begin{align*}
\lim_{\beta \to \infty}\frac{h_{(\beta+1)\l}}{h_{\beta\l}}= \lim_{\beta \to \infty}\frac{(\l/2)\sqrt{(2\beta+1)}}{(\l/2)\sqrt{(2\beta-1)}}= \lim_{\beta \to \infty}\frac{\sqrt{(2\beta+1)}}{\sqrt{(2\beta-1)}}=1.
\end{align*}
\end{proof}
\end{lemma}
\begin{lemma}
\label{lm:halfspace-circle-center}
For $k=\beta/(2(\beta-1))$ and $a,b\in \mathbb{R}^2$ ($a$ is fixed and $b$ is arbitrary), halfspace $\overrightarrow{a}.\overrightarrow{b}\leq (1/2k)\Vert a \Vert ^2$ passes through the center of disk $(b-ka)^2\geq(k\Vert a \Vert)^2$ if $\beta=2+\sqrt{2}$. 
\end{lemma}

\begin{proof}
It is enough to substitute $b$ in the given halfspace with $ka$ which is the center of the given disk.
\begin{align*}
&\overrightarrow{a}.\overrightarrow{(ka)}\leq \frac{1}{2k}\Vert a \Vert ^2 \Rightarrow k\Vert a \Vert^2 \leq \frac{1}{2k}\Vert a \Vert^2 \Rightarrow 2k^2\leq 1 \Rightarrow\\&2(\frac{\beta}{2(\beta-1)})^2 \leq 1 \Rightarrow -\beta^2+4\beta-2\leq 0 \Rightarrow \beta =2\pm \sqrt{2}
\end{align*}
Since $\beta\geq 1$, the $\beta=2+\sqrt{2}$ is valid. 
\end{proof}

\begin{table*}[!ht]
\begin{center}
\caption{Summary of experimental results}
\label{tbl:results-fitting-function}
\begin{tabular}{|l|l|l||l|l|l|l|}
\hline
x&$\HD=f(x)$&$d_E(x,\HD)$&Figure&$f(x,\SkD_{\infty})$&$d_E(x,\SkD_{\infty})$&Figure\\
\hline\hline
$\SkD_{1}$&$3.103x^2-0.013x+0.013$&$0.0016$&\ref{fig:1-h}&$0.942x+0.288$&$0.007$&\ref{fig:1-inf}\\
\hline
$\SkD_{2}$&$2.71x^2-0.48x+0.04$&$0.0019$&\ref{fig:2-h}&$0.858x+0.23$&$0.002$&\ref{fig:2-inf}\\
\hline
$\SkD_{3}$&$2.92x^2-0.82x+0.08$&$0.0021$&\ref{fig:3-h}&$0.86x+0.192$&$0.001$&\ref{fig:3-inf}\\
\hline
$\SkD_{1000}$&$4.26x^2-2.76x+0.47$&$0.0031$&\ref{fig:1000-h}&$0.999x+0.002$&$0.0$&\ref{fig:1000-inf}\\
\hline
$\SkD_{10000}$&$4.27x^2-2.78x+0.47$&$0.0031$&\ref{fig:10000-h}&$1.0x+0.0$&$0.0$&\ref{fig:10000-inf}\\
\hline
\end{tabular}
\end{center}
\end{table*}
\newpage
\begin{figure}[!ht]
  \centering
    \includegraphics[width=0.7\textwidth]{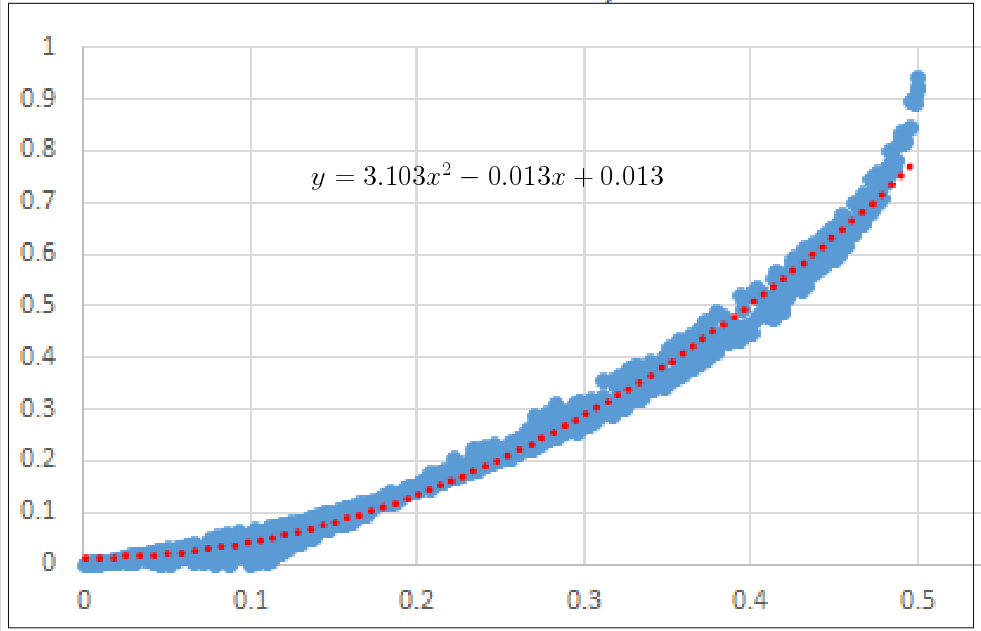}
  \caption{$\HD$ verses $\SkD_1$.}
  \label{fig:1-h}
\end{figure}

\begin{figure}[!ht]
  \centering
    \includegraphics[width=0.7\textwidth]{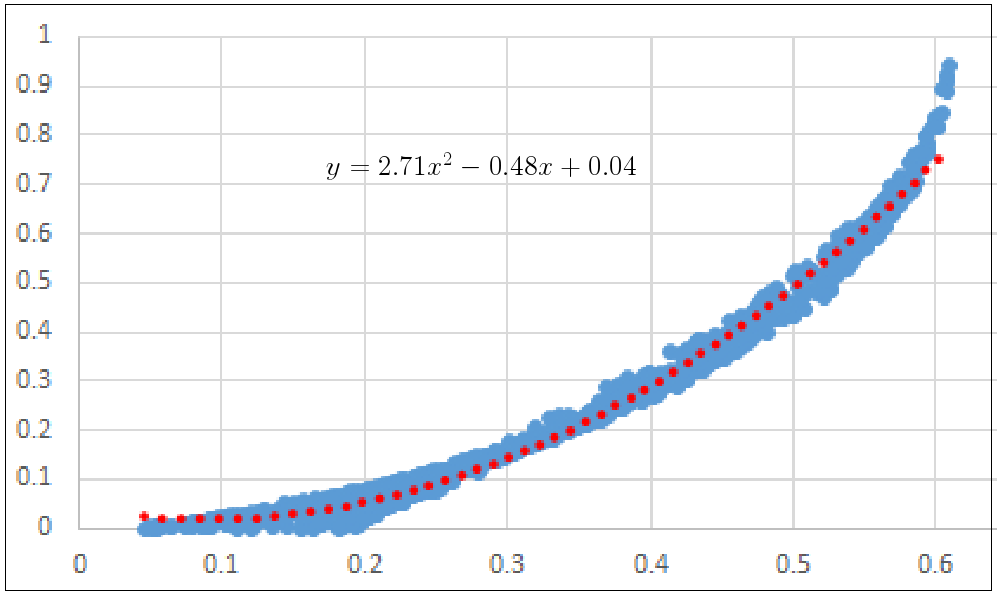}
  \caption{$\HD$ verses $\SkD_2$.}
  \label{fig:2-h}
\end{figure}

\begin{figure}[!ht]
  \centering
    \includegraphics[width=0.7\textwidth]{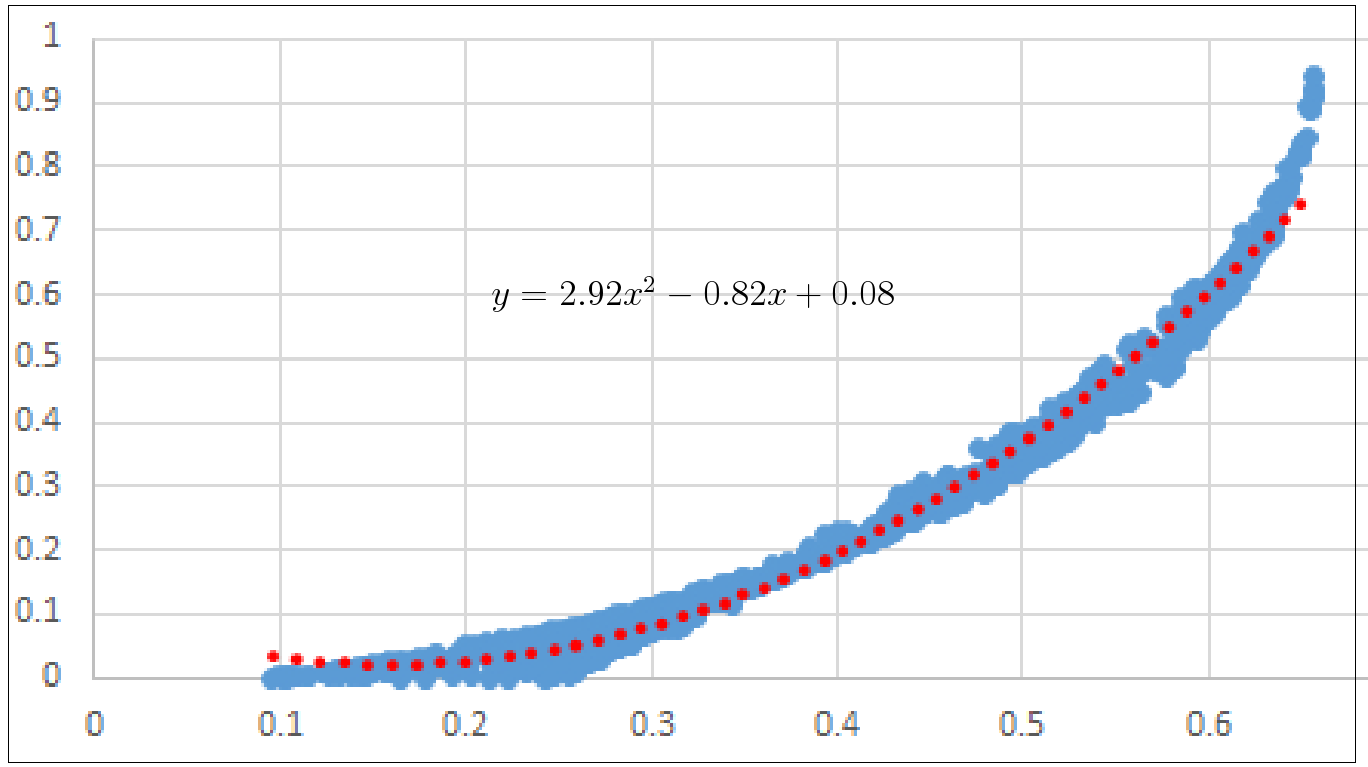}
  \caption{$\HD$ verses $\SkD_3$.}
  \label{fig:3-h}
\end{figure}

\begin{figure}[!ht]
  \centering
    \includegraphics[width=0.7\textwidth]{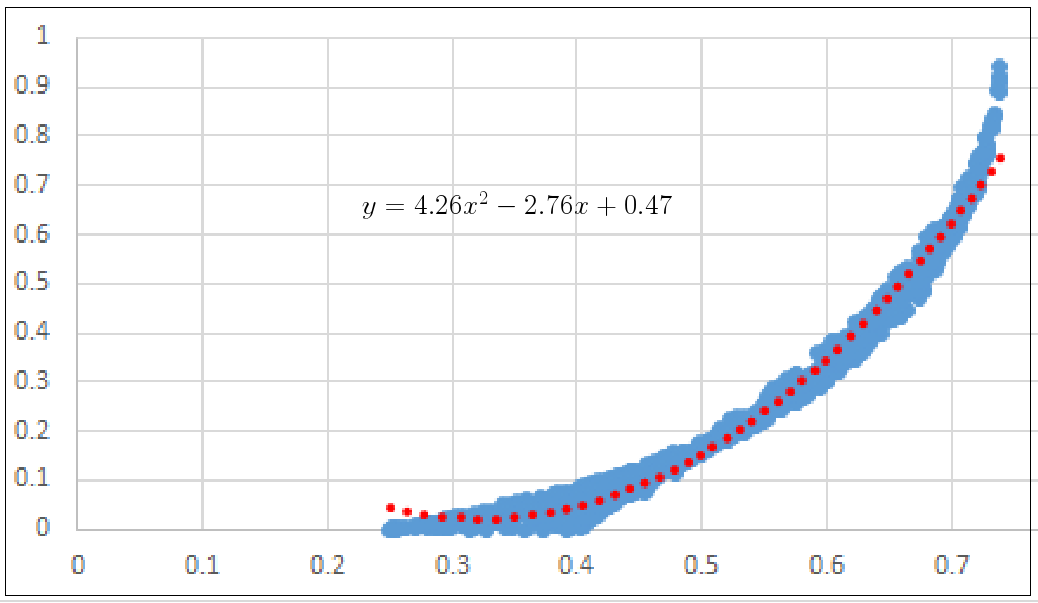}
  \caption{$\HD$ verses $\SkD_{1000}$.}
  \label{fig:1000-h}
\end{figure}

\begin{figure}[!ht]
  \centering
    \includegraphics[width=0.7\textwidth]{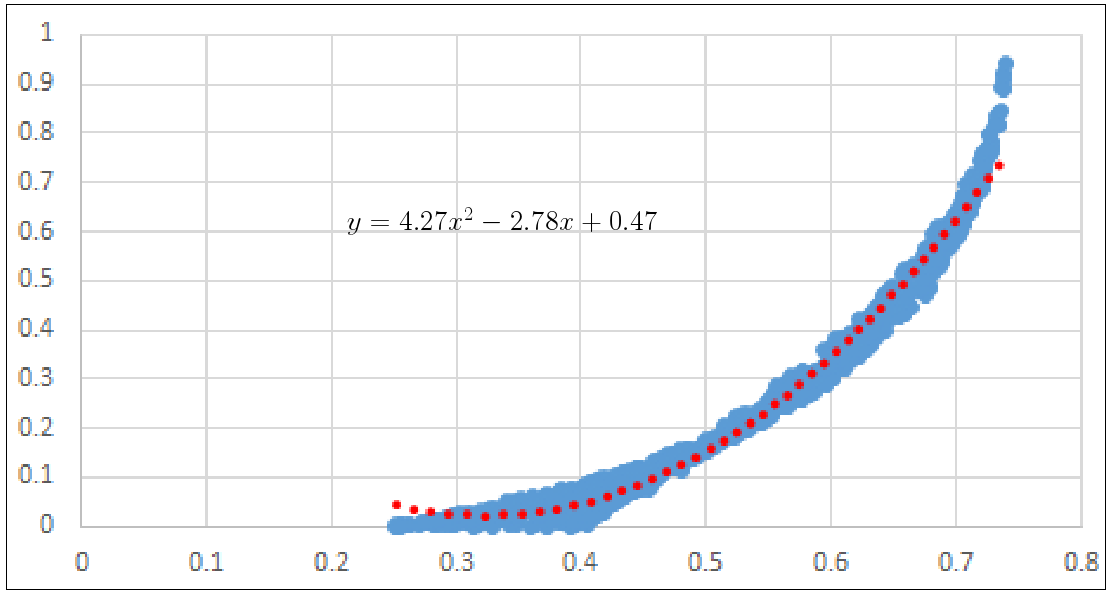}
  \caption{$\HD$ verses $\SkD_{10000}$.}
  \label{fig:10000-h}
\end{figure}

\begin{figure}[!ht]
  \centering
    \includegraphics[width=0.7\textwidth]{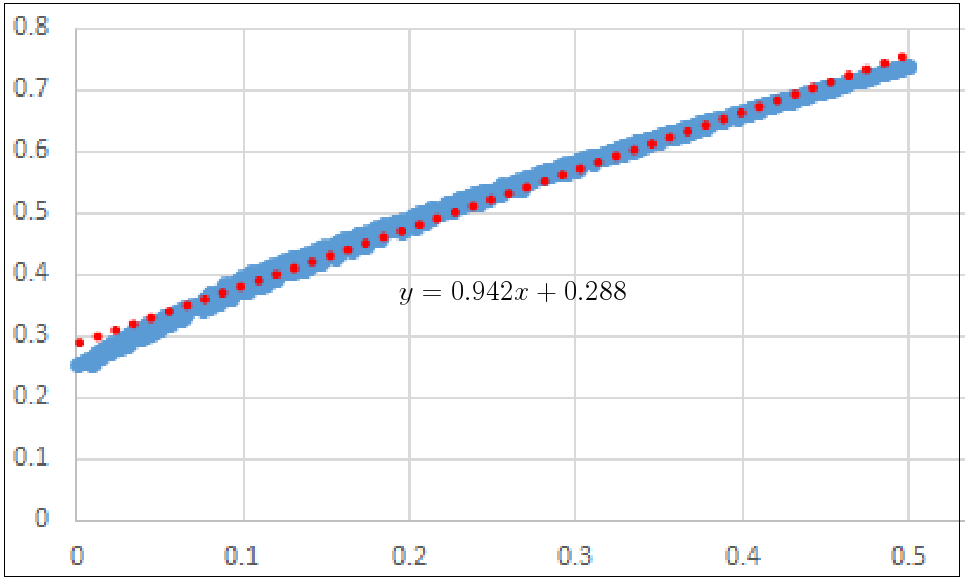}
  \caption{$\SkD_{\infty}$ verses $\SkD_1$.}
  \label{fig:1-inf}
\end{figure}

\begin{figure}[!ht]
  \centering
    \includegraphics[width=0.7\textwidth]{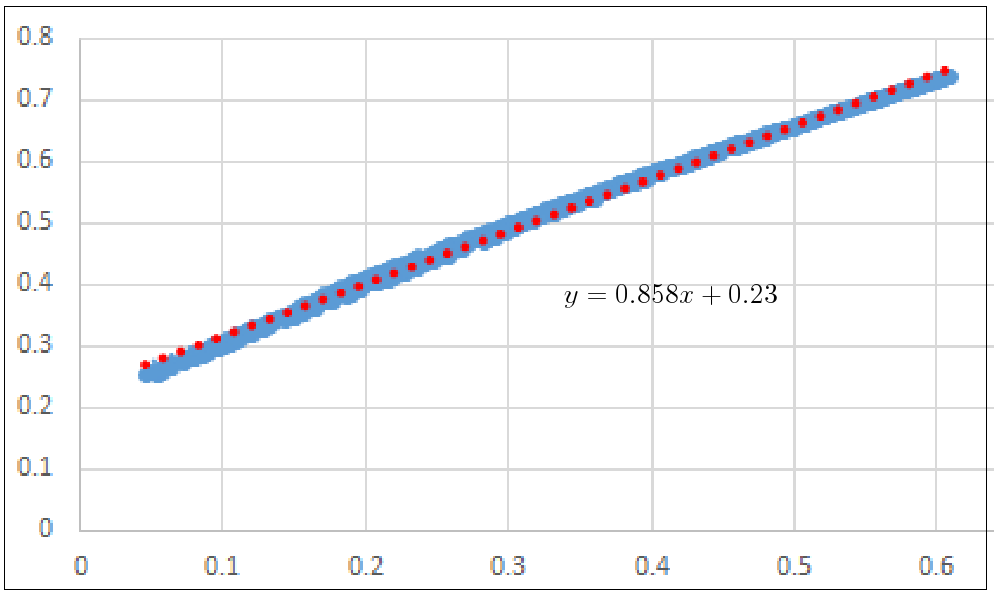}
  \caption{$\SkD_{\infty}$ verses $\SkD_2$.}
  \label{fig:2-inf}
\end{figure}

\begin{figure}[!ht]
  \centering
    \includegraphics[width=0.7\textwidth]{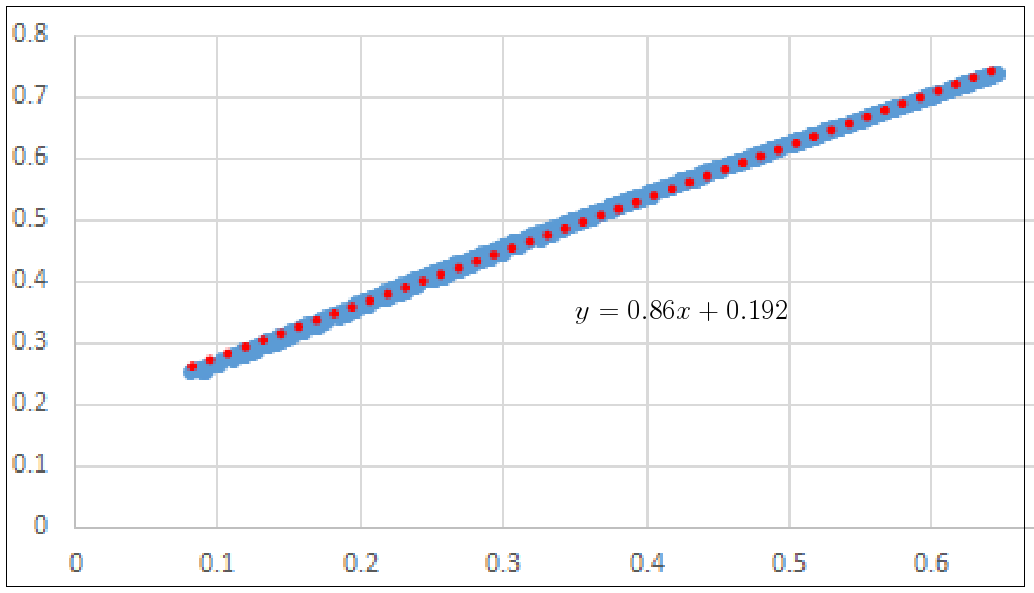}
  \caption{$\SkD_{\infty}$ verses $\SkD_3$.}
  \label{fig:3-inf}
\end{figure}

\begin{figure}[!ht]
  \centering
    \includegraphics[width=0.7\textwidth]{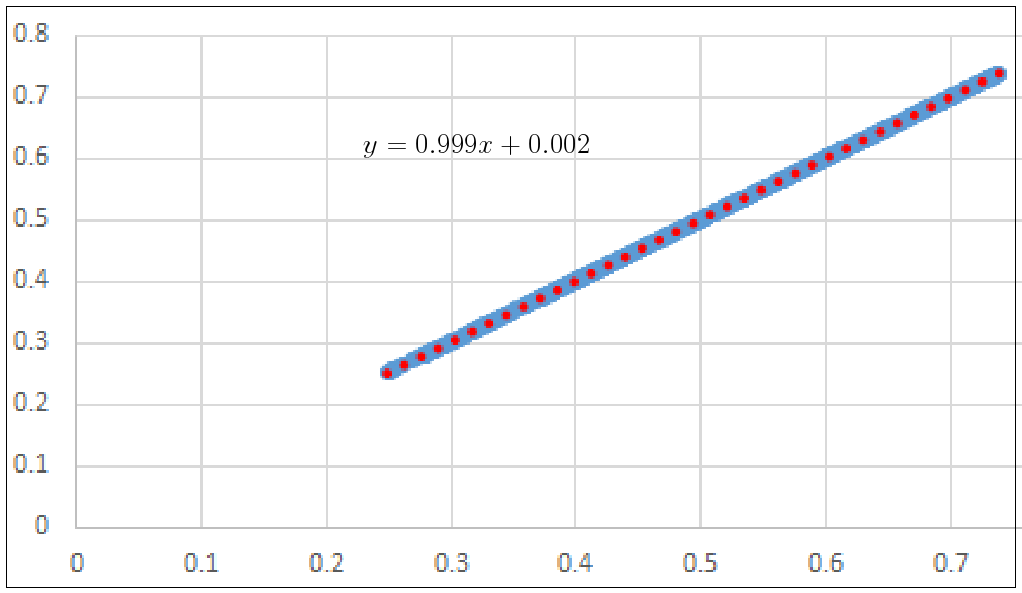}
  \caption{$\SkD_{\infty}$ verses $\SkD_{1000}$.}
  \label{fig:1000-inf}
\end{figure}

\begin{figure}[!ht]
  \centering
    \includegraphics[width=0.7\textwidth]{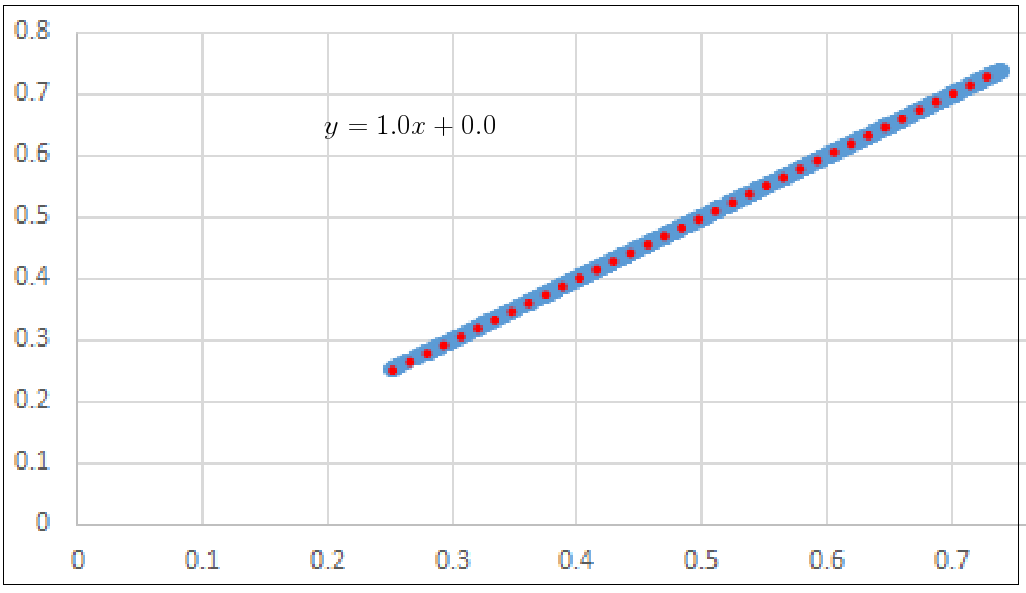}
  \caption{$\SkD_{\infty}$ verses $\SkD_{10000}$.}
  \label{fig:10000-inf}
\end{figure}

\begin{thebibliography}{9}

\bibitem{afshani2009approximate}
Peyman Afshani and Timothy~M Chan.
\newblock On approximate range counting and depth.
\newblock {\em Discrete \& Computational Geometry}, 42(1):3--21, 2009.

%%%%%%%%%%%%%%%%%%%%%2,3

\bibitem{aloupis2006geometric}
Greg Aloupis.
\newblock Geometric measures of data depth.
\newblock {\em DIMACS series in discrete mathematics and theoretical computer
  science}, 72:147, 2006.

%%%%%%%%%%%%%%%%5

\bibitem{aronov2008approximating}
Boris Aronov and Sariel Har-Peled.
\newblock On approximating the depth and related problems.
\newblock {\em SIAM Journal on Computing}, 38(3):899--921, 2008.

\bibitem{aronov2010approximate}
Boris Aronov and Micha Sharir.
\newblock Approximate halfspace range counting.
\newblock {\em SIAM Journal on Computing}, 39(7):2704--2725, 2010.

%%%8

\bibitem{bremner2008output}
David Bremner, Dan Chen, John Iacono, Stefan Langerman, and Pat Morin.
\newblock Output-sensitive algorithms for tukey depth and related problems.
\newblock {\em Statistics and Computing}, 18(3):259--266, 2008.

\bibitem{bremner2006primal}
David Bremner, Komei Fukuda, and Vera Rosta.
\newblock Primal-dual algorithms for data depth.
\newblock {\em DIMACS series in discrete mathematics and theoretical computer
  science}, 72:147, 2006.

\bibitem{chan2004optimal}
Timothy~M Chan.
\newblock An optimal randomized algorithm for maximum tukey depth.
\newblock In {\em Proceedings of the fifteenth annual ACM-SIAM symposium on
  Discrete algorithms}, pages 430--436. Society for Industrial and Applied
  Mathematics, 2004.

%%%12

\bibitem{chen2013absolute}
Dan Chen, Pat Morin, and Uli Wagner.
\newblock Absolute approximation of tukey depth: Theory and experiments.
\newblock {\em Computational Geometry}, 46(5):566--573, 2013.

%%%14

\bibitem{donoho1992breakdown}
David~L Donoho and Miriam Gasko.
\newblock Breakdown properties of location estimates based on halfspace depth
  and projected outlyingness.
\newblock {\em The Annals of Statistics}, pages 1803--1827, 1992.

\bibitem{elmore2006spherical}
Ryan~T Elmore, Thomas~P Hettmansperger, and Fengjuan Xuan.
\newblock Spherical data depth and a multivariate median.
\newblock {\em DIMACS Series in Discrete Mathematics and Theoretical Computer
  Science}, 72:87, 2006.

\bibitem{fattore2014measuring}
Marco Fattore, Rosanna Grassi, and Alberto Arcagni.
\newblock Measuring structural dissimilarity between finite partial orders.
\newblock In {\em Multi-indicator Systems and Modelling in Partial Order},
  pages 69--84. Springer, 2014.

%%%%18

\bibitem{har2011geometric}
Sariel Har-Peled.
\newblock {\em Geometric approximation algorithms}.
\newblock Number 173. American Mathematical Soc., 2011.

\bibitem{har2011relative}
Sariel Har-Peled and Micha Sharir.
\newblock Relative (p, $\varepsilon$)-approximations in geometry.
\newblock {\em Discrete \& Computational Geometry}, 45(3):462--496, 2011.

\bibitem{hotelling1990stability}
Harold Hotelling.
\newblock Stability in competition.
\newblock In {\em The Collected Economics Articles of Harold Hotelling}, pages
  50--63. Springer, 1990.

\bibitem{johnson1978densest}
David~S. Johnson and Franco~P Preparata.
\newblock The densest hemisphere problem.
\newblock {\em Theoretical Computer Science}, 6(1):93--107, 1978.

\bibitem{kirkpatrick1985framework}
David~G Kirkpatrick and John~D Radke.
\newblock A framework for computational morphology.
\newblock In {\em Machine Intelligence and Pattern Recognition}, volume~2,
  pages 217--248. Elsevier, 1985.

%%%24
\bibitem{liu2006data}
Regina~Y Liu.
\newblock {\em Data depth: robust multivariate analysis, computational
  geometry, and applications}, volume~72.
\newblock American Mathematical Soc., 2006.

\bibitem{liu1990notion}
Regina~Y Liu et~al.
\newblock On a notion of data depth based on random simplices.
\newblock {\em The Annals of Statistics}, 18(1):405--414, 1990.

\bibitem{liu2011lens}
Zhenyu Liu and Reza Modarres.
\newblock Lens data depth and median.
\newblock {\em Journal of Nonparametric Statistics}, 23(4):1063--1074, 2011.

%%%%28

\bibitem{oja1983descriptive}
Hannu Oja.
\newblock Descriptive statistics for multivariate distributions.
\newblock {\em Statistics \& Probability Letters}, 1(6):327--332, 1983.

\bibitem{rousseeuw1999regression}
Peter~J Rousseeuw and Mia Hubert.
\newblock Regression depth.
\newblock {\em Journal of the American Statistical Association},
  94(446):388--402, 1999.

\bibitem{rousseeuw1996algorithm}
Peter~J Rousseeuw and Ida Ruts.
\newblock Algorithm as 307: Bivariate location depth.
\newblock {\em Journal of the Royal Statistical Society. Series C (Applied
  Statistics)}, 45(4):516--526, 1996.

\bibitem{rousseeuw1998computing}
Peter~J Rousseeuw and Anja Struyf.
\newblock Computing location depth and regression depth in higher dimensions.
\newblock {\em Statistics and Computing}, 8(3):193--203, 1998.

\bibitem{shafer2012beginning}
Douglas~S Shafer and Z~Zhang.
\newblock Beginning statistics.
\newblock {\em Phylis-Barnidge publisher}, 304, 2012.

\bibitem{shahsavarifar2018computing}
Rasoul Shahsavarifar and David Bremner.
\newblock Computing the planar $\beta$-skeleton depth.
\newblock {\em arXiv preprint arXiv:1803.05970}, 2018.

\bibitem{shaul2011range}
Hayim Shaul.
\newblock {\em Range Searching: Emptiness, Reporting, and Approximate
  Counting}.
\newblock University of Tel-Aviv, 2011.

\bibitem{small1990survey}
Christopher~G Small.
\newblock A survey of multidimensional medians.
\newblock {\em International Statistical Review/Revue Internationale de
  Statistique}, pages 263--277, 1990.

\bibitem{struyf1999halfspace}
Anja~J Struyf and Peter~J Rousseeuw.
\newblock Halfspace depth and regression depth characterize the empirical
  distribution.
\newblock {\em Journal of Multivariate Analysis}, 69(1):135--153, 1999.

%%%38

\bibitem{tukey1975mathematics}
John~W Tukey.
\newblock Mathematics and the picturing of data.
\newblock In {\em Proceedings of the international congress of mathematicians},
  volume~2, pages 523--531, 1975.

\bibitem{yang2014depth}
Mengta Yang.
\newblock {\em Depth Functions, Multidimensional Medians and Tests of
  Uniformity on Proximity Graphs}.
\newblock PhD thesis, The George Washington University, 2014.

\bibitem{yang2017beta}
Mengta Yang and Reza Modarres.
\newblock $\beta$-skeleton depth functions and medians.
\newblock {\em Communications in Statistics-Theory and Methods}, pages 1--17,
  2017.

%%%42

\end{thebibliography}
\end{document}